\documentclass[twoside,a4paper]{article}
\usepackage{amsmath,graphicx,amssymb,fancyhdr,amsthm,enumerate,textcomp} 
\usepackage[usenames]{color}
\newtheorem{thm}{Theorem}[section]

\newtheorem{prop}[thm]{Proposition} 
 
\theoremstyle{definition} 
\newtheorem{defn}[thm]{Definition}
  
\theoremstyle{remark}  
\newtheorem{rem}[thm]{Remark}  
\def\beq{\begin{eqnarray}}  
\def\eeq{\end{eqnarray}}  
\def\bsp{\begin{split}}  
\def\esp{\end{split}}

\def\Tr{\mathrm{Tr}}  
\def\d{\mathrm{d}}

\newcommand{\mbold}[1]{\mbox{\boldmath{\ensuremath{#1}}}}

\begin{document}  
  
\title{\Large\textbf{Pseudo-Riemannian VSI spaces }}  
\author{{\large\textbf{Sigbj\o rn Hervik$^\text{\tiny\textleaf}$ and Alan Coley$^\heartsuit$} }    
 \vspace{0.3cm} \\    
$^\text{\tiny\textleaf}$Faculty of Science and Technology,\\    
 University of Stavanger,\\  N-4036 Stavanger, Norway     
\vspace{0.2cm} \\ 
$^\heartsuit$Department of Mathematics and Statistics,\\ 
Dalhousie University, \\ 
Halifax, Nova Scotia, Canada B3H 3J5 
\vspace{0.3cm} \\     
\texttt{sigbjorn.hervik@uis.no, aac@mathstat.dal.ca} }    
\date{\today}    
\maketitle  
\pagestyle{fancy}  
\fancyhead{} 
\fancyhead[EC]{S. Hervik \& A. Coley}  
\fancyhead[EL,OR]{\thepage}  
\fancyhead[OC]{Pseudo-Riemannian VSI metrics}  
\fancyfoot{} 
  
\begin{abstract} 
In this paper we consider pseudo-Riemannian spaces of
 arbitrary signature for which all of their polynomial curvature
 invariants vanish (VSI spaces).  We discuss an algebraic classification
 of pseudo-Riemannian spaces in terms of the boost weight decomposition
 and define the ${\bf S}_i$- and ${\bf N}$-properties, and show that if the
 curvature tensors of the space possess the ${\bf N}$-property then it is
 a VSI space.  We then use this result to construct a set of metrics that are
 VSI.  All of the VSI spaces constructed possess a geodesic, expansion-free,
 shear-free, and twist-free null-congruence.  We also discuss the related
 Walker metrics.  
 \end{abstract}

\section{Introduction}

Let us consider an arbitrary-dimensional pseudo-Riemannian space of signature $(k,k+m)$.
We will investigate when such a space can have a degenerate curvature structure;
in particular, we shall determine when all of its polynomial curvature invariants  vanish (VSI space).
Previously, the VSI spaces for Lorentzian metrics have been studied \cite{VSI} and it was shown that
these comprise a subclass of the degenerate Kundt metrics \cite{degen}.  Here, we will see that 
Kundt-like metrics also play a similar
role for pseudo-Riemannian VSI metrics of arbitrary signature.
In order to study such VSI metrics we will utilise curvature operators since any curvature
invariant is the trace of some curvature operator \cite{OP}.  We will 
use this fact to construct a set
of metrics with vanishing curvature invariants. All of the metrics constructed possess a
null-vector that is geodesic, shear-free, vorticity-free and expansion-free.

Let us first review the boost weight classification,
originally used to study degenerate metrics in Lorentzian geometry \cite{class},
in the pseudo-Riemannian case \cite{bw}.  The symmetry group of frame-rotations in
this case is $SO(k,k+m)$.  Any element $G$, can be
written $G=KAN$, where we have split it up into a compact spin piece, $K$, an Abelian boost piece,
$A$, and a piece consisting of null-rotations, $N$.  For $SO(k,k+m)$, $K\in SO(m)$, and there
are $k$-independent boosts (= the real rank of $SO(k,k+m)$).

Therefore, we first introduce a suitable null-frame such that the metric can be written:
\beq
\d s^2=2\left({\mbold\ell}^1{\mbold n}^1+\dots+{\mbold\ell}^I{\mbold n}^I+\dots+{\mbold\ell}^{k}{\mbold n}^{k}\right)+\delta_{ij}{\mbold m}^i{\mbold m}^j,
\eeq
where the indices $i=1,\dots, m$. The spins will act on ${\mbold m}^i$,  each boost will act on the pair of null-vectors, while the null-rotations will in general mix up null-vectors and spatial vectors. More precisely, 

\beq
\text{Spins:} && \tilde{\mbold\ell}^I={\mbold\ell}^I, ~\tilde{\mbold n}^I={\mbold n}^I, ~\tilde{\mbold m}^i=M^i_{~j}{\mbold m}^j, ~ (M^i_{~j})\in SO(m), \\
\text{Boosts:} && \tilde{\mbold\ell}^I=e^{\lambda_I}{\mbold\ell}^I, ~\tilde{\mbold n}^I=e^{-\lambda_I}{\mbold n}^I, ~\tilde{\mbold m}^i={\mbold m}^j, 
\eeq
while the null-rotations can be split up at each level. 
Considering the subset of forms $({\mbold\ell}^I,{\mbold n}^I,{\mbold\omega}^{\mu_I})$, 
where ${\mbold\omega}^{\mu_I}=\{{\mbold\ell}^{I+1},{\mbold n}^{I+1},\cdots,{\mbold\ell}^{k},{\mbold n}^{k},{\mbold m}^{i}\}$, then we can consider the $I$-th level null-rotations with respect to ${\mbold n}^I$:
\beq
\text{Null-Rot:} ~~ \tilde{\mbold\ell}^I={\mbold\ell}^I-z_{\mu_I}{\mbold\omega}^{\mu_I}-\frac 12z_{\mu_I}z^{\mu_I}{\mbold n}^I, ~\tilde{\mbold n}^I={\mbold n}^I, ~\tilde{\mbold\omega}^{\mu_I}={\mbold\omega}^{\mu_I}+z^{\mu_I}{\mbold n}^I, 
\eeq
and similarly for ${\mbold \ell}^I$. Note that there are $2(2k+m-2I)$ null-rotations at the $I$th level, 
making $2k(k+m-1)$ in total.  

\subsection{Boost weight decomposition}

Let us consider the $k$ independent boosts:
\beq
({\mbold \ell}^1,{\mbold n}^1)&\mapsto& (e^{\lambda_1}{\mbold\ell}^1,e^{-\lambda_1}{\mbold n}^1)\nonumber\\
({\mbold \ell}^2,{\mbold n}^2)&\mapsto& (e^{\lambda_2}{\mbold\ell}^2,e^{-\lambda_2}{\mbold n}^2)\nonumber\\
& \vdots & \nonumber\\
({\mbold\ell}^{k},{\mbold n}^{k})&\mapsto& (e^{\lambda_k}{\mbold\ell}^{k},e^{-\lambda_k}{\mbold n}^{k}).
\eeq
For a tensor $T$, we can then consider the boost weight of the components of this tensor, ${\bf b}\in \mathbb{Z}^k$, as follows. If the component $T_{\mu_1...\mu_n}$ transforms as:
\[
T_{\mu_1...\mu_n}\mapsto e^{-(b_1\lambda_1+b_2\lambda_2+...+b_k\lambda_k)}T_{\mu_1...\mu_n},
\]
then we will say the component $T_{\mu_1...\mu_n}$ is 
of boost weight ${\bf b}\equiv (b_1,b_2,...,b_k)$. We can now decompose a tensor 
into boost weights; in particular, 
\[ T=\sum_{{\bf b}\in  \mathbb{Z}^k}(T)_{\bf b},\] 
where $(T)_{\bf b}$ means the projection onto the components of boost weight ${\bf b}$. 

By considering tensor products, the boost weights obey the following additive rule: 
\beq
(T \otimes S)_{{\bf b}}=\sum_{\tilde{\bf b}+\hat{\bf b}={\bf b}}(T)_{\tilde{\bf b}}\otimes (S)_{\hat{\bf b}}.
\eeq

\subsection{The ${\bf S}_i$- and ${\bf N}$-properties}

Let us consider a tensor, $T$, and list a few conditions that the tensor components may fulfill:
\begin{defn} \label{cond}We define the following conditions:
\begin{enumerate}
\item[B1)]{} $(T)_{\bf b}=0$, for ${\bf b}=(b_1,b_2,b_3,...,b_k)$, $b_1>0$. 
\item[B2)]{} $(T)_{\bf b}=0$, for ${\bf b}=(0,b_2,b_3,...,b_k)$, $b_2>0$. 
\item[B3)]{} $(T)_{\bf b}=0$, for ${\bf b}=(0,0,b_3,...,b_k)$, $b_3>0$.
\item[$\vdots$]{} 
\item[B$k$)]{}  $(T)_{\bf b}=0$, for ${\bf b}=(0,0,...,0,b_k)$, $b_k>0$.
\end{enumerate}
\end{defn}

\begin{defn}
We will say that a tensor $T$ possesses the ${\bf S}_1$-property if and only if there exists a null frame such that condition B1) above is satisfied. Furthermore, we say that $T$ possesses the ${\bf S}_i$-property if and only if there exists a null frame such that conditions B1)-B$i$) above are satisfied.
\end{defn}
\begin{defn}
We will say that a tensor $T$ possesses the ${\bf N}$-property if and only if there exists a null frame such that conditions B1)-B$k$) in definition \ref{cond} are satisfied, \emph{and} 
\[ (T)_{\bf b}=0, \text{ for }  {\bf b}=(0,0,...,0,0).\] 

\end{defn}

\begin{prop}
For tensor products we have:
\begin{enumerate}
\item{}
Let $T$ and $S$ possess the ${\bf S}_i$- and ${\bf S}_j$-property, respectively. Assuming, 
with no loss of generality, that $i\leq j$, then $T\otimes S$ possesses the ${\bf S}_{i}$-property.
\item{} Let  $T$ and $S$ possess the ${\bf S}_i$- and ${\bf N}$-property, respectively. Then  $T\otimes S$  possesses the ${\bf S}_i$-property. If $i=k$, then   $T\otimes S$ possesses the ${\bf N}$-property.
\item{}  Let  $T$ and $S$ both possess the ${\bf N}$-property. Then  $T\otimes S$, and any contraction thereof, possesses the ${\bf N}$-property.
\end{enumerate}
\end{prop}

It is also useful to define a set of related conditions.  Consider a tensor, $T$, that does not necessarily
meet any of the conditions above.  However, since the boost weights ${\bf b}\in \mathbb{Z}^k\subset{\mathbb
R}^k$, we can consider a linear $GL(k)$ transformation, $G:\mathbb{Z}^k\mapsto \Gamma$, where $\Gamma$ is a
lattice in $\mathbb{R}^k$.  Now, if there exists a $G$ such that the transformed boost weights, $G{\bf b}$,
satisfy (some) of the conditions in Def.\ref{cond}, we will say, correspondingly, that $T$ possesses the ${\bf
S}^G_i$-property.  Similarly, for the ${\bf N}^G$-property.

If we have two tensors $T$ and $S$ both possessing the ${\bf S}_i^G$-property, with the same $G$, then when we take the tensor product: 
\[ (T\otimes S)_{G{\bf b}}=\sum_{G\hat{\bf b}+G\tilde{\bf b}=G{\bf b}}(T)_{G\hat{\bf b}}\otimes(S)_{G\tilde{\bf b}}.\]
Therefore, the tensor product will also possess the ${\bf S}_i^G$-property, with the same $G$.  This will be
useful later when considering degenerate tensors and metrics with degenerate curvature tensors.  Note
also that the ${\bf S}_i^G$-property reduces to the ${\bf S}_i$-property for $G=I$ (the identity).

\begin{rem}
A tensor, $T$, satisfying the ${\bf S}_i^G$-property or ${\bf N}^G$-property is not generically 
determined by its invariants in the sense that there may be another tensor, $T'$, with 
precisely the same invariants. The ${\bf S}_i$-property thus implies a certain \emph{degeneracy} in the tensor. 
\end{rem}

\subsection{The VSI properties}

We can now state two results that will be useful for us. 

\begin{thm}
A pseudo-Riemannian space is VSI if and only if all the curvature operators are nilpotent. 
\end{thm}
\begin{proof}
The proof of this follows almost directly from paper \cite{OP}. In particular, consider a curvature operator ${\sf T}$. Then:
\beq
&\text{All invariants  }\Tr({\sf T}^k)\text{ are zero.}&\nonumber \\
&\Updownarrow& \nonumber\\
&\text{All eigenvalues of }{\sf T}\text{ are zero.}&\nonumber \\
&\Updownarrow& \nonumber\\
&{\sf T}\text{ is nilpotent.}&\nonumber
\eeq
Since any polynomial curvature invariant is the trace of some curvature operator, the theorem follows. 
\end{proof}
So if a curvature operator is nilpotent then we know that all its invariants are zero. The next result gives us a necessary criterion for when an operator is nilpotent.
\begin{thm}
If a even-ranked tensor $T$ possesses the ${\bf N}^G$-property, then the operator ${\sf T}$ (obtained by raising/lowering indices) is nilpotent.
\end{thm}
\begin{proof}
This result follows from an equivalent proof to that in \cite{OP}, which we will include here for completeness. 
We note that since $T$ possesses the ${\bf N}^G$-property, so do $T\otimes T$, $T\otimes T\otimes T$, etc. 
Moreover, since the metric ${\bf g}$ is of boost weight zero, the operators ${\sf T}^k$ also possess the  
${\bf N}^G$-property. The coefficients of the eigenvalue equation of ${\sf T}$ consist of 
traces of ${\sf T}^k$ which necessarily has zero boost weight, thus we get $\Tr({\sf T}^k)=0$ and all eigenvalues are consequently zero. This implies that ${\sf T}$ is nilpotent.
\end{proof}

\section{Invariant null planes: Walker metrics} 
A set of vectors ${\mbold \ell}^I$, $I=1,...,k'$, is said to span a $k'$-dimensional invariant plane 
if the $k'$-vector, ${\mbold\ell}^1\wedge...\wedge{\mbold\ell}^{k'}$, is \emph{recurrent}. Using the dual one-forms, with components $\ell^I_\mu$, the recurrent requirement is equivalent to:
\[ \left(\ell^1_{[\mu_1}\ell^2_{\mu_2}...\ell^{k'}_{\mu_{k'}]}\right)_{;\nu}=
\ell^1_{[\mu_1}\ell^2_{\mu_2}...\ell^{k'}_{\mu_{k'}]}k_{\nu},\]
where $k_{\nu}$ is an arbitrary vector.
We can interpret this as requiring that the volume form of the invariant plane is recurrent. 

If, in addition, the vectors ${\mbold\ell}^I$, $I=1,...,k'$, are all null
and mutually orthogonal, then the invariant plane is \emph{null}.
Henceforth, we assume that the vectors span a $k'$-dimensional invariant null
plane, $\mathcal{N}$.  Consider the orthogonal complement,
$\mathcal{N}_\bot$.  Then, since the volume form of $\mathcal{N}_\bot$ is
the Hodge dual of the volume form of $\mathcal{N}$, the complement
$\mathcal{N}_\bot$ is also invariant.  Note that since $\mathcal{N}$ is
null, we have $\mathcal{N}\subset\mathcal{N}_\bot$.

If the pseudo-Riemannian space is a \emph{Walker metric}, which is defined as a space admitting a $k'$-dimensional invariant null plane ${\cal N}$,
then it can be shown that the set of null vectors ${\mbold\ell}_{I}\subset \mathcal{N}$ 
are such that 
\[ [{\mbold\ell}_{I},{\mbold\ell}_{J}] \subset \mathcal{N}.\]
 By Frobenius' theorem, these vectors then span a submanifold for which there 
 exists an adapted set of coordinates, $(v^1,...,v^{k'})$ spanning the null-plane, 
 and a complimentary set, $(u^1,...,u^{k'},x^i)$. Furthermore, we can choose the 
 set of null-vectors ${\mbold\ell}_I$ so that ${\mbold\ell}_I=\frac{\partial}{\partial v^I}$, 
 and their dual one-forms are ${\mbold\ell}^I=\mathrm{d}u^I$. 

Let us choose one of these null-forms, say ${\mbold\ell}^1=\ell_1\d u^1=\d u^1$, and form the null frame $\{ {\mbold\ell}^1,{\mbold n}^1,{\mbold\omega}^i\}$. Since this is a null-frame, the corresponding connection (spin) coefficients fulfill the fundamental antisymmetry:
\beq
\Gamma_{\mu\nu\alpha}=-\Gamma_{\nu\mu\alpha},
\eeq
where $\Gamma_{\mu\nu\alpha}\equiv g_{\mu\rho}\Gamma^\rho_{~\nu\alpha}$.
Furthermore, using the Frobenius theorem for both $\mathcal{N}$ and $\mathcal{N}_\bot$ we get:
\beq
\d {\mbold\ell}^1=0,\quad \d {\mbold\omega}^i=-\frac 12c^i_{jk}{\mbold\omega}^j\wedge{\mbold\omega}^k-c^i_{1j}{\mbold\ell}^1\wedge{\mbold\omega}^j.
\eeq
 Using the relation, $\d {\mbold\omega}^\alpha=-\Gamma^\alpha_{[\mu\nu]}{\mbold\omega}^\mu\wedge{\mbold\omega}^\nu$, 
 then implies the vanishing of certain connection coefficients. For example, consider:
\beq
{\ell}^\mu\nabla_\mu\ell_\nu=-\Gamma_{0\nu 0}=-\Gamma_{00\nu}=0.
\eeq
Hence, $\ell_{\mu}$ is geodesic. Furthermore, since $\d {\mbold\ell}^1=0$, $\ell_\mu$ is vorticity-free. 
This implies we can write the covariant derivative as: 
\beq
\ell_{\mu;\nu}=L_{11}\ell_\mu\ell_\nu+L_{1i}\ell_{(\mu}\omega^i_{~\nu)}+L_{ij}\omega^i_{~\mu}\omega^j_{~\nu},
\eeq
where the $\omega^i_{~\mu}$ are the components of ${\mbold\omega}^i$. Defining the extrinsic curvature and the vorticity: 
\[ K^{ij}=L^{(ij)}=\ell_{(\mu;\nu)}\omega^{i\mu}\omega^{j\nu}, \quad  A^{ij}=L^{[ij]}=\ell_{[\mu;\nu]}\omega^{i\mu}\omega^{j\nu},\] 
we note that $A_{ij}=0$ (vorticity-free). Calculating $K_{ij}$:
\[ K_{ij}=-\Gamma_{0ij}-\Gamma_{0ji}=\Gamma_{i0j}+\Gamma_{j0i}=\Gamma_{ij0}+\Gamma_{ji0}=0.\]
Therefore, $K_{ij}=0$ and $\ell^1$ is expansion-free and shear-free. Consequently, we have:
\beq
\ell_{\mu;\nu}=L_{11}\ell_\mu\ell_\nu+L_{1i}\ell_{(\mu}\omega^i_{~\nu)}.
\eeq
The Walker metrics 
thus allows for a null-vector ${\mbold\ell}^1$ which is geodesic, expansion-free, shear-free and vorticity-free. 
In fact, we note that the above metric possesses $k'$ null-vectors ${\mbold\ell}^I$ which are geodesic, vorticity-free, shear-free and expansion-free. This is reminicient of the Kundt form of the metric in Lorentzian geometry. Indeed, this leads us to define pseudo-Riemannian Kundt metrics in analogy to the Lorentzian case. 
Indeed, these pseudo-Riemannian Kundt metrics will have the Walker metrics as a special case (however, not all pseudo-Riemannian Kundt metrics will be Walker metrics).

 Furthermore, Walker \cite{Walker}  showed that the requirement of an 
invariant $k'$-dimensional null plane implies that the (Walker) metric can be written in the canonical form: 
\beq
\mathrm{d}s^2=\d u^I(2\delta_{IJ}\d v^J+B_{IJ}\d u^J+H_{Ii}\d x^i)+A_{ij}\d x^i \d x^j,
\eeq
where $B_{IJ}$ is a symmetric matrix which may depend on all of the coordinates, while $H_{Ii}$ and $A_{ij}$ does not depend on the coordinates $v^I$. Note that if the signature is $(k,k+n)$, then $k'\leq k$. If $k'=1$, then there exists an invariant null-line, while if $k'=k$, then the invariant null-plane is of maximal dimension.

\section{Pseudo-Riemannian Kundt metrics}

In the Lorentzian case ($k=1$) the Kundt metrics play an important role
for degenerate metrics, and VSI metrics in particular \cite{VSI}.  Here we will argue
that their pseudo-Riemannian analogues also play an important role for
pseudo-Riemannian spaces of arbitary signature.

Let us define the pseudo-Riemannian Kundt metrics in a similar fashion, namely: 
\begin{defn}
A pseudo-Riemannian Kundt metric is a metric which possesses a non-zero null vector ${\mbold\ell}$ which is geodesic, expansion-free, twist-free and shear-free.
\end{defn}

We will consequently consider metrics of the form
\beq 
\d s^2=2\d u\left[\d v+H(v,u,x^C)\d u+W_{A}(v,u,x^C)\d x^A\right]+{g}_{AB}(u,x^C)\d x^A\d x^B
\label{Kundt}
\eeq
The metric (\ref{Kundt})
possesses a null vector field ${\mbold\ell}$ obeying\footnote{If, in addition $L_{1i}=0$, the vector $\ell_\mu$ is also recurrent, and if $L_{1i}=L_{11}=0$, then $\ell_\mu$ is covariantly constant. }
\[ \ell_{\mu;\nu}=L_{11}\ell_\mu\ell_\nu+L_{1i}\ell_{(\mu}\omega^i_{~\nu)},\]
and consequently, 
\beq
\ell^{\mu}\ell_{\nu;\mu}=\ell^{\mu}_{~;\mu}=\ell^{\nu;\mu}\ell_{(\nu;\mu)}=\ell^{\mu;\nu}\ell_{[\mu;\nu]}=0;
\eeq 
i.e., it is geodesic, non-expanding, shear-free and
non-twisting. If this is a pseudo-Riemannian space of signature $(k,k+m)$, then the transversal metric 
\[ \d s^2_{1}={g}_{AB}(u,x^C)\d x^A\d x^B,\] 
will be of signature $(k-1,k-1+m)$. 

By using a normalised frame we can calculate the boost weight ${\bf b}=(1,b_2,...,b_k)$ and $(0,b_2,...,b_k)$ components of the Riemann tensor: 
\beq
R_{\hat 0\hat 1\hat 0\hat A}&=&-\frac 12 W_{\hat A,vv}, \\
R_{\hat 0\hat 1\hat 0\hat 1}&=& -H_{,vv}+\frac 14\left(W_{\hat{A},v}\right)\left(W^{\hat A,v}\right), \\
R_{\hat 0\hat 1\hat A\hat B}&=& W_{[\hat A}W_{\hat B],vv}+W_{[\hat A;\hat B],v}, \\
R_{\hat 0\hat A\hat 1\hat B}&=& \frac 12\left[-W_{\hat B}W_{\hat A,vv}+W_{\hat A;\hat B,v}-\frac 12 \left(W_{\hat A,v}\right)\left(W_{\hat B,v}\right)\right], \\
R_{\hat A\hat B\hat C\hat D}&=&\tilde{R}_{\hat A\hat B\hat C\hat D}.
\eeq

Note that the Riemann tensor satisfies the ${\bf S}_1$-property if 
$R_{\hat 0\hat 1\hat 0\hat A}=-\frac 12 W_{\hat A,vv}=0$. Now, whether the Riemann 
tensor satisfies any of the other requirements depends on the components of boost weight ${\bf b}=(0,b_2,...,b_k)$: 
\beq
H_{,vv}-\frac 14\left(W_{\hat A,v}\right)\left(W^{\hat A,v}\right) &=& \sigma, \\
W_{[\hat A;\hat B],v} &=& {\sf a}_{\hat A\hat B}, \\
W_{(\hat A;\hat B),v}-\frac 12 \left(W_{\hat A,v}\right)\left(W_{\hat B,v}\right) &=& {\sf s}_{\hat A\hat B},
\label{Wcsi4}\eeq
and the components $\tilde{R}_{\hat A\hat B\hat C\hat D}$.

\subsection{The ${\bf N}$-property and VSI spaces}
We can now consider the conditions which are required for these pseudo-Riemannian Kundt metrics 
to be VSI metrics. 

To get a set of necessary conditions we consider the Riemann tensor, $R$ as an operator ${\sf
R}:\wedge^2T_pM\mapsto \wedge^2T_pM$, as follows:  ${\sf R}=(R^{\alpha\beta}_{\phantom{\alpha\beta}\mu\nu})$.
We note that for the pseudo-Riemannian Kundt metrics, if $ R_{\hat 0\hat 1\hat 0\hat A}=-\frac 12 W_{\hat
A,vv}=0$, the question of when ${\sf R}$ possesses the ${\bf N}$-property is related to the values of the components
$\sigma$, ${\sf a}_{\hat A\hat B}$, ${\sf s}_{\hat A\hat B}$ and $\tilde{R}_{\hat A\hat B\hat C\hat D}$.

Therefore, a sufficient criterion for the Riemann curvature operator to have only zero polynomial invariants is that $\sigma=0$ and the matrices 
\[ \left[{\sf a}^{\hat A}_{\phantom{\hat A}\hat B}\right], \quad \left[{\sf s}^{\hat A}_{\phantom{\hat A}\hat B}\right],\quad  \left[{\widetilde R}^{\hat A\hat B}_{\phantom{\hat A\hat B}\hat{C}\hat{D}}\right], \]
have only zero-eigenvalues; hence, they are \emph{nilpotent}.  
\subsection{Nested Kundt metrics}

Consider now the case where the transverse metric $\d
s_{1}={g}_{AB}(u,x^C)\d x^A\d x^B$ of signature $(k-1,k-1+m)$ is also
pseudo-Riemannian Kundt.  This means that it can also be written in Kundt
form (where $u$ is considered as a parameter).  The transverse space will
now also be independent of one of the null-coordinates.  Again a special
case would be when this transverse space is also pseudo-Riemannian Kundt.

Therefore, let $(v^{a_I},u^{b_I})$, $a_I,b_I=1,...,I$ be null-coordinates, and $x^{\mu_I}$, $\mu_I=2I,..,k+m$ be ``transversal'' coordinates. We can then consider the $I$-th nested Kundt metric:
\beq
\d s_I^2=2\d{u}^I\left[\d v^I+H^I(v^I,u^a,x^{\rho_I})\d u^I+W^I_{\mu_I}(v^I,u^a,x^{\rho_I})\d x^{\mu_I}\right]+{g}^I_{\mu_I\nu_I}(u^a,x^{\rho_I})\d x^{\mu_I}\d x^{\nu_I},
\label{KundtI}
\eeq
for which the boost weight component ${\bf b}=(0,..0,1,b_{I},...,b_k)$ is $-{\frac{1}{2}}W^I_{\hat {\mu}_I,v^Iv^I}$. Similarly, the boost weight ${\bf b}=(0,...,0,b_{I},...,b_k)$ components can be written in terms of: 
\beq
H^I_{,v^Iv^I}-\frac 14\left(W^I_{\hat \mu_I,v^I}\right)\left(W^{I~\hat \mu_I,v^I}\right) &=& \sigma^I, \\
W^I_{[\hat \mu_I;\hat \nu_I],v^I} &=& {\sf a}^I_{\hat \mu_I\hat \nu_I}, \\
W^I_{(\hat \mu_I;\hat \mu_I),v^I}-\frac 12 \left(W^I_{\hat \mu_I,v^I}\right)\left(W^I_{\hat \nu_I,v^I}\right) &=& {\sf s}^I_{\hat \mu_I\hat \nu_I},
\label{Wcsi4I}\eeq
and the components $\tilde{R}^I_{\hat \mu_I\hat \nu_I\hat \rho_I\hat \lambda_I}$. 
\section{Pseudo-Riemannian Kundt VSI metrics}

We can now construct VSI metrics from the nested Kundt metrics by requiring that 
they fulfill the ${\bf N}$-property, and then iteratively solving the above equations for the components. We can assume the metric is an $I$-th order nested Kundt metric and solve for the metric components $H^I$ and $W^I_i$. 

We will start out by considering the lower-dimensional cases.  In 4 dimensions (4D), the only case not
considered in detail is the Neutral $(2,2)$-signature case.  In 5 dimensions, the only case not consider
previously in the literature is the $(2,3)$-signature case.  These two cases will be considered next.  We
will also give a class of VSI metrics of $(k,k+m)$-signature.

\subsection{4D Neutral case}
Let us consider the 4D neutral case. Here we can write 
\beq
\d s^2=2\left({\mbold\ell}^1{\mbold n}^1+{\mbold\ell}^2{\mbold n}^2\right). 
\eeq
We will consider the pseudo-Riemannian Kundt case and at the 1st level the transverse space  is 2-dimensional. Conseqently, there is only one independent component of the Riemann tensor, namely $\widetilde{R}_{3434}$. This is of boost weight ${\bf b}=(0,0)$ and so, requiring the ${\bf N}$-property, this must be flat space. Therefore, we can write: 
\[ 2{\mbold\ell}^2{\mbold n}^2=2\d u^2\d v^2=-\d T^2+\d X^2.\] 
Based on the previous discussion, we can find two classes of pseudo-Riemannian Kundt VSI metrics. They can be written: 
\beq
\d s^2=2\d u^1\left(\d v^1+H\d u^1+W_{\mu_1}\d x^{\mu_1}\right)+2\d u^2\d v^2,
\eeq 
where:
\paragraph{Null case:}
\beq
W_{\mu_1}\d x^{\mu_1}&=& v^1W^{(1)}_{u^2}(u^1,u^2)\d u^2+W^{(0)}_{u^2}(u^1,u^2,v^2)\d u^2+W^{(0)}_{v^2}(u^1,u^2,v^2)\d v^2,\nonumber \\
H&=& v^1 H^{(1)}(u^1,u^2,v^2)+H^{(0)}(u^1,u^2,v^2),
\eeq
\paragraph{Spacelike/timelike case:}
\beq
W_{\mu_1}\d x^{\mu_1}&=& v^1W^{(1)}\d X+W^{(0)}_T(u^1,T,X)\d T+W^{(0)}_X(u^1,T,X)\d X, \nonumber\\
H&=& \frac{(v^1)^2}8{\left(W^{(1)}\right)^2}+v^1H^{(1)}(u^1,T,X)+H^{(0)}(u^1,T,X),
\eeq
and 
\beq
W^{(1)}=-\frac{2\epsilon}{X}, \text{ where } \epsilon=0, 1.
\eeq

We note that these possess an invariant null-line if $W^{(1)}=0$, and a 2-dimensional invariant null-plane if $W^{(0)}_{v^2}=0$ for the null 
case\footnote{In order for the spacelike/timelike case to possess an invariant 
null 2-plane, it needs to be a special case of the null case.}.
\subsection{The $(2,3)$-signature case}
Considering the 5D (2+3)-signature case,  we can write 
\beq
\d s^2=2\left({\mbold\ell}^1{\mbold n}^1+{\mbold\ell}^2{\mbold n}^2\right)+(\mbold m^1)^2. 
\eeq
Again considering the pseudo-Riemannian Kundt case, we see that the 1st level 
pseudo-Riemannian Kundt has 3-dimensional transverse space. At the 2nd level, 
the transverse space is 1-dimensional. We can thus choose ${\mbold\omega}^1=(\d x)^2$. At the 2nd level, we therefore get the standard Lorentzian Kundt VSI spaces: 
\beq
\d s^2_1=2\d v^2\left(\d v^2+H^2\d u^2+W^2\d x\right)+(\d x)^2 
\eeq
where
\beq
H^2&=&\frac{(v^2)^2}{8}\left(W^{2,(1)}\right)^2+v^2H^{2,(1)}(u^1,u^2,x)+H^{2,(0)}(u^1,u^2,x), \label{5D:H2}\\
W^2_{\mu_2}\d x^{\mu_2}&=&v^2W^{2,(1)}\d x+W^{2,(0)}(u^1,u^2,x)\d x,\label{5D:W2}
\eeq
where 
\beq
W^{2,(1)}=-\frac{2\epsilon}{x}.
\eeq
Then using these metrics and solving the remaining equations for the 5D Kundt metric gives us several cases. The metric can be written:
\beq
\d s^2&=&2\d u^1\left(\d v^1+H^1\d u^1+W^1_{\mu_1}\d x^{\mu_1}\right)\nonumber \\
&& +2\d u^2\left(\d v^2+H^2\d u^2+W^2\d x\right)+(\d x)^2, 
\eeq
where the functions $H^2$ and $W^2$ are given by eqs. (\ref{5D:H2}) and (\ref{5D:W2}). 

\paragraph{Case 1, $\epsilon=1$:}
\beq
H^1&=& {(v^1)^2}H^{1,(2)}+
(v^1)H^{1,(1)}(u^1,u^2,v^2,x)+H^{1,(0)}(u^1,u^2,v^2,x), \nonumber \\
W^1_{\mu_1}\d x^{\mu_1}&=&v^1{\bf W}^{1,(1)}+W^{1,(0)}_{\mu_1}(u^1,u^2,v^2,x)\d x^{\mu_1},
\eeq
where, using the definition $(u^1,v^1,u^2,v^2)=(u,v,U,V)$, 
\beq
{\bf W}^{1,(1)}&=&-\frac{2\d V}{V+Cx+Dx^2}+\frac{2[V-Dx^2]\d x}{[V+Cx+Dx^2]x} \\
&& +\frac{2[DV+x(CD-C_{,U})+x^2(D^2-D_{,U})]\d U}{V+Cx+Dx^2}, \nonumber \\
H^{1,(2)}&=& -\frac{x[2(CD-C_{,U})+x(D^2-2D_{,U})]}{2[V+Cx+Dx^2]^2}, 
\eeq
where $C=C(u,U)$ and $D=D(u,U)$ are arbitrary functions. We note that 
\[ {\bf W}^{1,(1)}=\d\phi, \quad \phi=2\ln x-2\ln[V+Cx+Dx^2]+2\int D(u,U)\d U.\]
\paragraph{Case 2, $\epsilon=0$, null:}
 
\beq
W^1_{\mu_1}\d x^{\mu_1}&=& v^1W^{1,(1)}_{u^2}(u^1,u^2,{x})\d {u^2}+W^{1,(0)}_{\mu_1}(u^1,u^2,v^2,x)\d x^{\mu_1}, \nonumber \\
H^1&=& v^1H^{1,(1)}(u^1,u^2,v^2,x)+H^{1,(0)}(u^1,u^2,v^2,x).
\eeq
We note that these possess an invariant null-line if $\epsilon=0$ (null case), and $W^{1,(1)}=0$. The metric possesses a 2-dimensional invariant null plane if $\epsilon=0$, $W^{1,(0)}_{v^1}=W^{1,(0)}_{v^2}=0$ and $\partial_{v^2}W^{1,(0)}_x=0$.
\subsection{A class of $(k,k+m)$-signature VSI metrics. }
\label{arbVSI}
We can now give a class of metrics that are VSI and are the generalisations of of the null-case mentioned previously. 
We write them as: 
\beq
\d s^2=\sum_{I=1}^k2{\mbold\ell}^{I}{\mbold n}^{I}+\sum_{i,j=1}^m\delta_{ij}{\mbold m}^i{\mbold m}^j,
\eeq
where 
\beq
{\mbold\ell}^{I}&=&\d u^I, \\
{\mbold n}^{I}&=&\d v^I+\left[v^IH^{I,(1)}(u^{a},v^{a_I},x^i)+H^{I,(0)}(u^{a},v^{a_I},x^i)\right]\d u^I\nonumber \\
& & +\left[v^IW^{I,(1)}(u^a,v^{a_{I+1}},x^i)\d u^{I+1}+W^{I,(0)}_{\mu_I}(u^a,v^{a_I},x^i)\d x^{\mu_I}\right], \\
{\mbold m}^i&=& \d x^i,
\eeq
and where the indices have the following ranges:
\beq
&&a= 1,2,...,k; \quad 
a_I= I+1,...,k; \quad i=1,...,m,\nonumber \\
&& \d x^{\mu_I}=\{\d u^{a_I},\d v^{a_I},\d x^i\}.
\eeq
In the appendix a sketch of a proof is given which shows that these spaces are indeed VSI. 

\section{A class of Ricci-flat metrics}
Let us consider the special case where the only non-zero functions are $H^{I,(0)}$:
\beq
{\mbold\ell}^{I}=\d u^I, \quad
{\mbold n}^{I}=\d v^I+H^{I,(0)}(u^{a},v^{a_I},x^i)\d u^I,\quad
{\mbold m}^i= \d x^i.
\eeq
We note that these are all Walker metrics possessing an invariant $k$-dimensional null plane. 
The only non-zero components of the Ricci tensor are:
\[ R_{u^Iu^I}=-\Box H^{I,(0)},\] where
\[ \Box=\sum_{J=1}^k2\left(\frac{\partial}{\partial u^J}-H^{J,{(0)}}\frac{\partial}{\partial v^J}\right)\frac{\partial}{\partial v^J}+\sum_{i}\left(\frac{\partial}{\partial {x^i}}\right)^2.\]

We note that in general the Ricci operator, ${\sf R}=(R^\mu_{~\nu})$, is 2-step nilpotent; i.e., ${\sf R}^2=0$. In order for this metric to be Ricci-flat, we need to solve the $k$ equations:
\beq
\Box H^{I,(0)}=0.
\eeq
Since the functions $H^{I,(0)}$ do not depend on $v^J$, for $J=1,...,I$, these equations simplify and we have:
\beq
\Box H^{I,(0)}=\left[\sum_{J=I+1}^k2\left(\frac{\partial}{\partial u^J}-H^{J,{(0)}}\frac{\partial}{\partial v^J}\right)\frac{\partial}{\partial v^J}+\sum_{i}\left(\frac{\partial}{\partial {x^i}}\right)^2\right]H^{I,(0)}=0.
\eeq
We can therefore start solving the equation for $I=k$, which reduces to solving the 
Laplacian over the space $\delta_{ij}\d x^i\d x^j$. The function $H^{I,(0)}$ can then 
be used in the equation for  $H^{I-1,(0)}$, which again can (at least in principle) be solved. Then this can again be inserted into the equation for  $H^{I-2,(0)}$, etc. One can therefore systematically solve these equations since the equation for  $H^{I,(0)}$ only involves the functions  $H^{I+1,(0)}$,  $H^{I+2,(0)}$, etc. 
\subsection{Neutral 4D space}
Consider first the Neutral 4D case. Here we have, $H^2=0$; consequently, there 
is only one differential equation that needs to be solved: 
\beq
\Box H^{1,(0)}(u^1,u^2,v^2)=\frac{\partial^2}{\partial u^2\partial v^2}H^{1,(0)}(u^1,u^2,v^2)=0.
\eeq
This can be solved in the standard way:
\beq
H^{1,(0)}(u^1,u^2,v^2)=f(u^1,u^2)+g(u^1,v^2),
\eeq
for arbitrary functions $f$ and $g$. Consequently, we get the standard Ricci-flat neutral metrics: 
\beq
\d s^2=2\d u^1\left[\d v^1+(f(u^1,u^2)+g(u^1,v^2))\d u^1\right]+2\d u^2\d v^2.
\eeq 
This is also, by construction, a VSI space. 
\subsection{The $(2,3)$-signature case.}
Here, requiring Ricci flatness, we get the two equations:
\beq
\Box H^{1,(0)}(u^1,u^2,v^2,x)&=&\left(\frac{\partial^2}{\partial u^2\partial v^2}+\frac{\partial^2}{\partial x^2}\right)H^{1,(0)}(u^1,u^2,v^2,x)=0,\nonumber \\
\Box H^{2,(0)}(u^1,u^2,x)&=&\frac{\partial^2}{\partial x^2}H^{2,(0)}(u^1,u^2,x)=0, 
\eeq
These two equations, which are the Laplacians in flat 1- and 3-dimensional space, respectively, 
can easily be solved.

\section{Discussion}
In this paper we have considered the class of VSI spaces in arbitrary signature spaces. 
By using a boost weight classification of tensors we were able to give a set of necessary 
criteria for a space to be VSI. Then we constructed a set of spaces that had this property and,
 consequently, are VSI spaces. These spaces were all of the form of a pseudo-Riemannian Kundt spacetime 
 which is a generalisation of the Lorentzian Kundt metrics. We also discussed the related Walker metrics.

However, there are still a few open questions:
\begin{enumerate}
\item{} \emph{Do all VSI metrics possess the ${\bf N}^G$-property?}  A stronger question is:
\item{} \emph{Do all VSI metrics possess the ${\bf N}$-property?}  
\item{} \emph{Are all VSI metrics of pseudo-Riemannian Kundt form?} 
\end{enumerate}
In the Lorentzian case, the answer to these questions are `yes' (questions 1 and 2 are 
equivalent in Lorentzian case), but for other signatures there is the possibility for other types of VSI metrics. 

This paper therefore sets the stage for further investigations on VSI metrics of spaces of arbitrary signature.

\section*{Acknowledgements} 
The main part of this work was done during a visit to Dalhousie
University April-June 2010 by SH. The work was supported by
NSERC of Canada (AC) and by a Leiv Eirikson mobility grant from the  
Research Council of Norway, project no: {\bf 200910/V11} (SH). 
\appendix

\section{Sketch of the VSI proof}

Here we will give a sketch of a proof that the metrics in section \ref{arbVSI}  are indeed VSI. Therefore,
consider the metric given by eqn. (42). 
The frame is given by:
\beq
{\mbold\ell}^{I}&=&\d u^I, \\
{\mbold n}^{I}&=&\d v^I+\left[v^IH^{I,(1)}(u^{a},v^{a_I},x^i)+H^{I,(0)}(u^{a},v^{a_I},x^i)\right]\d u^I\nonumber \\
& & +\left[v^IW^{I,(1)}(u^a,v^{a_{I+1}},x^i)\d u^{I+1}+W^{I,(0)}_{\mu_I}(u^a,v^{a_I},x^i)\d x^{\mu_I}\right], \\
{\mbold m}^i&=& \d x^i.
\eeq
We will investigate the Ricci rotation coefficients, $\Gamma^{\mu}_{~\alpha\beta}$. By the Cartan structural equations, these can be deduced by the coefficients of the exterior derivative of the null-frame:
\[ \d {\mbold \omega}^\mu=-\Gamma^{\mu}_{~\alpha\beta}{\mbold\omega}^{\alpha}\wedge{\mbold\omega}^\beta.\] 
Let us use the above null-frame and check the boost weight of the connection coefficient 
$\Gamma^\mu_{~\alpha\beta}$. Recall that an index ${\mbold\ell}^I$ downstairs (or ${\mbold n}^I$ upstairs)
 has an associated $-1$ boost weight, while an index ${\mbold n}^I$ downstairs 
 (or ${\mbold \ell}^I$ upstairs) has an associated $+1$ boost weight. We will, in short, refer to the boost weight as a triple $(\ast,\ast,\ast)$ which corresponds to the indices $(a<I,a=I,a>I)$.

We note that $\d {\mbold \ell}^I=0$ and $\d {\mbold \omega}^i=0$. All the 
linearly independent connection coefficients are therefore obtained from $\d {\mbold n}^I$:
\beq
\d {\mbold n}^{I}&=&\underbrace{H^{I,(1)}(u^{a},v^{a_I},x^i)\d v^I\wedge {\mbold\ell}^I}_{(0,-1,0)+(0,-2,\ast)}\nonumber \\
&& +\underbrace{\left[v^I\d H^{I,(1)}(u^{a},v^{a_I},x^i)+\d H^{I,(0)}(u^{a},v^{a_I},x^i)\right]\wedge {\mbold\ell}^I}_{(-1,-2,0)+(0,-2,\ast)}\nonumber \\
&& +\underbrace{W^{I,(1)}(u^a,v^{a_{I+1}},x^i)\d v^I\wedge{\mbold \ell}^{I+1}}_{{(0,0,-1)+(0,-1,-1)+(0,-2,-1)}}\nonumber \\
&&+
\underbrace{v^I\d \left[ W^{I,(1)}(u^a,v^{a_{I+1}},x^i)\right]\wedge{\mbold\ell}^{I+1}}_{(0,-1,-1)}+\underbrace{\d \left[W^{I,(0)}_{\mu_I}(u^a,v^{a_I},x^i)\d x^{\mu_I}\right]}_{(0,-1,\ast)}\nonumber
\eeq
We see that all the coefficients fulfill the ${\bf N}$-property; consequently, 
\emph{the connection coefficients $\Gamma^\mu_{~\alpha\beta}$ satisfy the ${\bf N}$-property}. 

We can see from the previous equations that the Riemann tensor, $R$, also fulfills the ${\bf N}$-property. The covariant derivative of a tensor $T$ can, symbolically, be written as: 
\[ \nabla T=\partial T-\sum\Gamma\star T.\] 
The second term of the right-hand side is algebraic, so if $T$ and $\Gamma$ both fulfill the ${\bf N}$-property, so will the term $\sum\Gamma\star T$. Therefore, we only have to check the partial derivatives, $\partial T$. 

For the Riemann tensor, we can now use the Bianchi identity and the generalised Ricci identity 
to show that all covariant derivatives of the Riemann tensor  also fulfill the ${\bf N}$-property. Alternatively, since the dangerous terms that can \emph{potentially} make the derivatives violate the ${\bf N}$-property are the partial derivatives with respect to $v^I$, we can carefully keep track of the $v^I$-dependence of each of the components. We then notice that boost weight terms $(0,-1,\ast)$ will either be independent of $v^I$, or linear in $v^I$. However, the terms that are linear in $v^I$ are of boost weight $(0,-1,-1)$ and independent of $v^{I+1}$. Therefore,  $\partial T$ will still obey the ${\bf N}$-property.

\end{document}